\documentclass[a4paper,USenglish]{lipics-v2016}
 
\usepackage{microtype}

\usepackage{xifthen} 
\usepackage{xspace}

\newtheorem{claim}[theorem]{Claim}

\newcommand\p{\mbox{\bf P}\xspace}
\newcommand\np{\mbox{\bf NP}\xspace}

\newcommand\dtime{\mbox{\bf DTIME}}

\newcommand{\eqdef}{\mathrel{\mathop:}=}
\newcommand{\OPT}{\ensuremath{\mathrm{OPT}}\xspace}
\newcommand{\opt}{\OPT\xspace}
\newcommand{\poly}{\ensuremath{\mathrm{poly}}\xspace}

\newcommand{\dcs}{\textsc{DCS}\xspace}
\newcommand{\dcsf}{\textsc{Densest Common Subgraph}\xspace}
\newcommand{\ma}{\textsc{DCS-MA}\xspace}
\newcommand{\am}{\textsc{DCS-AM}\xspace}
\renewcommand{\aa}{\textsc{DCS-AA}\xspace}
\newcommand{\mm}{\textsc{DCS-MM}\xspace}
\newcommand{\kma}{\textsc{D\ensuremath{k}CS-MA}\xspace}
\newcommand{\minrep}[1][]{\ifthenelse{\isempty{#1}}{\textsc{MinRep}}{\ensuremath{#1}-\textsc{MinRep}}\xspace}
\newcommand{\labelcover}{\textsc{Label Cover}\xspace}
\newcommand{\score}{\ensuremath{\mathrm{score}}\xspace}
\newcommand{\dks}{\textsc{D}\ensuremath{k}\textsc{S}\xspace}
\newcommand{\dksf}{\textsc{Densest }\ensuremath{k}\textsc{-Subgraph}\xspace}
\newcommand{\dsf}{\textsc{Densest Subgraph}\xspace}
\newcommand{\mis}{\textsc{MIS}\xspace}
\newcommand{\misf}{\textsc{Maximum Independent Set}\xspace}
\newcommand{\mcss}{\textsc{MCSS}\xspace}
\newcommand{\mcssf}{\textsc{Minimum Common Spanning Subgraph}\xspace}
\newcommand{\setcover}{\textsc{Set Cover}\xspace}
\newcommand{\ekvc}{\textsc{E}\ensuremath{k}\textsc{VC}\xspace}
\newcommand{\ekvcf}{\textsc{E}\ensuremath{k}\textsc{-Vertex Cover}\xspace}

\title{On Finding Dense Common Subgraphs}
\titlerunning{On Finding Dense Common Subgraphs} 

\author[1]{Moses Charikar}
\author[2]{Yonatan Naamad}
\author[3]{Jimmy Wu}
\affil[1]{Computer Science Department, Stanford University, Stanford, USA\\
  \texttt{moses@cs.stanford.edu}}
\affil[2]{Amazon.com, Inc., Palo Alto, USA.\thanks{This work was done while the author was at the Department of Computer Science, Princeton University.}\\
  \texttt{ynaamad@amazon.com}}
\affil[3]{Computer Science Department, Stanford University, Stanford, USA\\
  \texttt{jimmyjwu@cs.stanford.edu}}
\authorrunning{M. Charikar, Y. Naamad, and J. Wu} 

\Copyright{Moses Charikar, Yonatan Naamad, and Jimmy Wu}


\begin{document}

\maketitle

\begin{abstract}
We study the recently introduced problem of finding dense common subgraphs:
Given a sequence of graphs that share the same vertex set, the goal is to find a subset of vertices $S$ that maximizes some aggregate measure of the density of the subgraphs induced by $S$ in each of the given graphs.
Different choices for the aggregation function give rise to variants of the problem that were studied recently.
We settle many of the questions left open by previous works, showing \np-hardness, hardness of approximation, non-trivial approximation algorithms, and an integrality gap for a natural relaxation.
\end{abstract}

\section{Introduction}
\label{sec:intro}

We study the recently introduced problem of finding dense common subgraphs ({\bf DCS}):
Given a sequence of $T$ graphs (sometimes called snapshots or frames) that share the same vertex set $V$, the goal is to find a subset of vertices $S \subseteq V$ that maximizes some aggregate measure of the density of the subgraphs induced by $S$ in each of the given graphs.
Different choices for the aggregation function give rise to different problem variants; two notable ones are (1) {\bf DCS-MA}, where the goal is to maximize the {\bf M}inimum (over the frames) of the {\bf A}verage degree in the induced subgraph, and (2) {\bf DCS-AM}, where the goal is to maximize the {\bf A}verage (over the frames) of the {\bf M}inimum degree in the induced subgraph.

Note that the problem for a single frame ($T=1$) is essentially the \dsf problem:
Given a graph, find an induced subgraph that maximizes the ratio of edges to vertices.
This is a classical optimization problem with well known polynomial time flow-based and LP-based algorithms~\cite{charikar2000greedy, goldberg1984finding}.
Thus the problems we study are natural generalizations of this maximum density subgraph objective to sequences of graphs.

The \dcs problem was recently introduced by Jethava and Beerenwinkel~\cite{jethava2015finding}; specifically, they study the \ma variant. For this, they proposed a linear program (a generalization of the LP for \dsf) and a greedy algorithm.
They showed, numerically, that the LP is not optimal, and also showed the the greedy algorithm is not optimal, but gave no approximation guarantees. They conjecture that the \ma problem is \np-hard.
Later followup work by Andersson et al.~\cite{anderssonfinding} ran experiments with the greedy algorithm and described a Lagrangian relaxation of the LP that could be solved more efficiently.
In very recent work of Semertzidis et al.~\cite{semertzidis2016best}, the authors explored four different variants of the \dcs problem corresponding to different choices of the aggregation function over the frames.
For the {\bf DCS-MM} (Min Min) and {\bf DCS-AA} (Average Average) problems, they gave simple exact algorithms.
They also proposed algorithms for \ma and \am, but only prove lower bounds on their approximation ratio.
Around the same time, Galimberti et al.~\cite{galimberti2017core} gave an exponential-time 2-approximation algorithm for \ma.

\dcs-like problems have appeared in various other research communities. In network biology, Hu et al.~\cite{hu2005mining} studied a problem with the same input as \dcs (that is, a graph sequence) but with a different objective from any of the four above, in order to understand the function of gene clusters. In machine learning, Jethava et al.~\cite{jethava2013lovasz} showed connections between support vector machines (SVMs) and a problem which is nearly identical to \ma.

In this work, we initiate a systematic theoretical study of the \ma and \am problems.
Prior to this, there were no nontrivial hardness results or efficient approximation algorithms for either of these problems. Moreover, we believe that this perspective, in which a classic combinatorial optimization problem is given a temporal or time-like dimension, deserves broader theoretical exploration. Such problems are not only interesting in their own right, but also arise naturally in social network analysis (as \dcs does), computational biology, and other important application domains. As an example of another classically tractable problem that becomes more interesting in the multi-frame case, we study in Appendix \ref{sec:mcss} the approximability of a natural generalization of the \textsc{Minimum Spanning Tree} problem. We anticipate that it will be interesting to explore temporal generalizations of other classic problems.

\subsection{Our Results}

For \ma, we give approximation algorithms with ratio $O(\sqrt{n \log T})$ (where $T$ is the number of graphs in the sequence) and $O(n^{2/3})$ irrespective of $T$. The first bound is better when $T = 2^{O(n^{1/6})}$, and in particular when $T$ is polynomial in $n$.

Additionally, we show an integrality gap of $\Omega(n/\log n)$ for a natural linear programming relaxation, introduced by \cite{jethava2015finding} and \cite{anderssonfinding}. This formalizes and significantly strengthens an experimental observation made by those authors that the LP is inexact.

On the complexity side, we prove that \ma is at least as hard to approximate as \minrep, a well-studied minimization version of \labelcover, and therefore cannot be approximated to within a factor of $2^{\log^{1-\epsilon} n}$ unless $\np \subseteq \dtime\left(n^{\mathrm{polylog } n}\right)$. This resolves a question left open by \cite{jethava2015finding} and \cite{semertzidis2016best}. Furthermore we show that, assuming a recent popular conjecture concerning the hardness of planted instances of the \dksf problem, \ma cannot be approximated to within a factor of $n^{1/4-\epsilon}$, and that even for $T=2$ frames, it cannot be approximated to within $n^{1/8-\epsilon}$. Finally in Appendix \ref{sec:kma}, we show that $f(n)$-hardness for (worst-case) \dksf implies $O(\sqrt{f(n)})$-hardness for D$k$CS-MA, a parameterized variant of \ma studied by \cite{semertzidis2016best}.

For the \am problem, we prove \np-hardness of approximation to within a factor of $n^{1-\epsilon}$, via reduction from \misf. This essentially matches the trivial upper bound of $n$.

Despite this hardness, we show that \am becomes tractable for small values of $T$. In particular, it can be solved exactly in time $n^T \cdot \text{poly}(n,T)$, and even has a fixed-parameter FPTAS: for every $\epsilon>0$, it can be $(1+\epsilon)$-approximated in time $f(T) \cdot \text{poly}(n, \epsilon^{-1})$.

\section{Preliminaries}
\label{sec:preliminaries}

In the \dcsf (\dcs) problem, we are given a sequence of graphs $\mathcal{G} = (G_1,G_2,\ldots,G_T)$ on the same vertex set $V$, and we must find a subset $S \subseteq V$ that maximizes the \emph{aggregate density} of the subgraphs induced by $S$. Different definitions of aggregate density give rise to different variants of \dcs:

\begin{itemize}
    \item (\mm) $\min_{i \in [T]} \text{min-deg}(G_i[S])$
    
    Here, $\text{min-deg}(G_i[S])$ is the minimum induced degree, i.e. $\min_{v \in S} \text{deg}_{G_i[S]}(v)$.
    
    \item (\ma) $\min_{i \in [T]} \left(\sum_{v \in S} \text{deg}_{G_i[S]}(v)\right)/{|S|}$
    
    This definition emphasizes the induced degrees of the vertices, but note that this is equivalent to $\min_{i \in [T]} |E(G_i[S])|/{|S|}$.
    
    \item (\am) $\sum_{i \in [T]} \text{min-deg}(G_i[S])$
    
    \item (\aa) $\sum_{v \in S} \left(\sum_{v \in S} \text{deg}_{G_i[S]}(v)\right)/{|S|}$
\end{itemize}

All four variants were studied by Semertzidis et al. \cite{semertzidis2016best}, who show that \mm can be solved by a simple greedy procedure and \aa easily reduces to the classic \dsf problem. They also present heuristics for \ma and \am, but do not provide approximation upper bounds; and they conjecture (but do not prove) that both problems are \np-hard. \ma was also studied by \cite{jethava2015finding} and \cite{anderssonfinding}, who experiment with linear programming relaxations and other heuristics.

Some of our hardness results are obtained by reducing from \minrep, which we now review.

\begin{definition}[\minrep]
We are given a bipartite graph $G = (A,B,E)$ and a partition of both sides into $A = A_1 \cup \cdots \cup A_k$ and $B = B_1 \cup \cdots \cup B_k$; each part $A_i$ or $B_j$ is called a \emph{supervertex}. We say there is a \emph{superedge} $(A_i,B_j)$ iff there are any edges in $A_i \times B_j$, and that a pair of vertices $a \in A_i, b \in B_j$ \emph{cover} $(A_i,B_j)$ iff $(a,b) \in E$. The goal is then to pick sets of vertices $A' \subseteq A, B' \subseteq B$ of minimum total size $|A'|+|B'|$ such that all superedges are covered by some pair of vertices in $A' \times B'$.
\end{definition}

As shown in \cite{kortsarz2001hardness}, \minrep has the following gap hardness:

\begin{theorem}[from \cite{kortsarz2001hardness}]
For every constant $\epsilon > 0$, the following promise problem cannot be solved in polynomial time unless $\np \subseteq \dtime\left(n^{\mathrm{polylog } n}\right)$: Given a \minrep instance $G = (A=\bigcup_{i=1}^k A_i, B=\bigcup_{i=1}^k B_i, E)$, distinguish between the following cases:
\begin{itemize}
    \item (YES instance) There exists a labeling $A' \subseteq A, B' \subseteq B$ of size $|A'|+|B'| = 2k$.
    \item (NO instance) Every labeling has size at least $2k \cdot 2^{\log^{1-\epsilon} n}$.
\end{itemize}
\end{theorem}

\section{Algorithms and Hardness for \ma}
\label{sec:ma}

\subsection{Approximating \ma}
\label{sec:approxma}

In this section, we first present a $O(\sqrt{n \log T})$-approximation algorithm for \ma. Notably, this simplifies to  $O(\sqrt{n \log n})$ when $T = \poly(n)$. We then show how to augment the algorithm so that the approximation ratio never surpasses $O(n^{2/3})$ even when $T$ is allowed to be super-polynomial in $n$.

\begin{theorem}
There exists an $O(\sqrt{n \log T})$-approximation algorithm for \ma.
\label{thm:maalg1}
\end{theorem}

The algorithm returns the better of two feasible solutions. The first feasible solution it considers is that containing all vertices, i.e. $S = V$. The second feasible solution is constructed greedily. For a set $S' \subset V$ of vertices, we say that a frame $G_i$ is covered if $E_i[S']$ contains at least one edge. Now initialize $S' \leftarrow \emptyset$. As long as some graphs remain uncovered, we add to $S'$ the two vertices $u, v \in V$ such that $S' \cup \{u, v\}$ covers as many graphs as possible (that is, $u$ and $v$ induce at least as many edges among the uncovered graphs as does any other pair). When all graphs are covered, $S = S'$ forms our second feasible solution.

Recall that the objective we seek to maximize is $\score(S) \eqdef \min_{i \in [T]} \frac{|E(G_i[S])|}{|S|}$. Let $k \eqdef |\opt|$; although this value is unknown to us, we can analyze, in terms of $k$, the quality of the two feasible solutions:

\begin{lemma}
If $S = V$, then $\score(S)$ is an $n/k$ approximation to $\score(\opt)$.
\label{lem:allv}
\end{lemma}

\begin{proof} 
\begin{equation*}
    \score(V)
    = \frac{\min_i |E_i[V]|}{|V|}
    \geq \frac{\min_i |E_i[\opt]|}{|V|}
    = \frac{\min_i |E_i[\opt]|}{|OPT|} \cdot \frac{|\opt|}{|V|}
    = \score(\opt) \cdot \tfrac{k}{n}.
\end{equation*}
\end{proof}

\begin{lemma}
If $S$ is constructed as in the above greedy algorithm, then $\score(S)$ is a $2 k \ln T$ approximation to $\score(\opt)$.
\label{lem:greedyma}
\end{lemma}

\begin{proof} 
Consider an intermediate state $i$ of our algorithm, with partial solution $S'_i$ covering all but $t_i$ of the frames. Let $\mathcal{T}_i$ be the set of $t_i$ subgraphs that \opt induces on each of the uncovered frames. Since the average degree in each subgraph in $\mathcal{T}_i$ is at least $\score(\opt)$, there are at least $\frac{k \cdot \score(\opt)}{2}$ induced edges in each frame of $\mathcal{T}_i$, and, summing over all these frames, $\frac{t_i \cdot k \cdot \score(\opt)}{2}$ edges total induced by $\opt$. As there are $\binom{k}{2}$ vertex pairs in $\opt$, at least one such pair induces an edge in at least $\frac{t_i k \cdot \score(\opt)}{2 \cdot \binom{k}{2}}$ of them. Thus, the next vertex pair chosen by the greedy algorithm covers at least a $\frac{k \cdot \score(\opt)}{2 \cdot \binom{k}{2}} \geq \frac{\score(\opt)}{k}$ fraction of the previously-uncovered frames. In general, the number of uncovered frames satisfies $t_i \leq \left(1 - \frac{\score(\opt)}{k}\right)t_{i-1}$ with $t_0 = T$. Therefore, $t_i \leq T \left(1 - \frac{\score(\opt)}{k} \right)^i$, and in particular $t_i < 1$ when $i \geq \frac{k \ln T}{\score(\opt)}$. Thus, the algorithm halts after at most $\frac{k \ln T}{\score(\opt)}$ iterations.

As each iteration adds at most two new vertices to $S'$, the ultimate size of $S'$ is at most $\frac{2k \ln T}{\score(\opt)}$. Combining this with the fact that the returned solution covers at least one vertex in each frame, we have

\[
    \score(S)
    = \frac{\min_i E_i[S]}{|S|}
    \geq \frac{1}{\frac{2k \ln T}{\score(\opt)}}
    = \frac{\score(\opt)}{2k \ln T}.
\]
\end{proof}

With these lemmas in place, Theorem \ref{thm:maalg1} is a simple corollary.

\begin{proof}[Proof of Theorem \ref{thm:maalg1}]
Let $\alpha$ be the score of $V$, and let $\beta$ be the score of the solution generated by the aforementioned greedy algorithm. By Lemmas \ref{lem:allv} and \ref{lem:greedyma}, $\alpha \geq \score(\OPT) \cdot \tfrac{k}{n}$ and $\beta \geq \tfrac{\score(\opt) }{2k \ln T}$. Thus, the better of the two solutions has score at least

\[
    \max(\alpha, \beta) = \sqrt{\max(\alpha, \beta)^2} \geq \sqrt{\alpha\cdot\beta} = \sqrt{\score(\opt) \cdot \frac{k}{n}\cdot\frac{\score(\opt) }{2k \ln T}} = \frac{\score(\opt) }{\sqrt{2n \ln T}}.
\]
\end{proof}

While the above algorithm provides an $O(\sqrt{n \log n})$ approximation when $T = \poly(n)$, it can be substantially worse when $T$ is allowed to be super-polynomial in $n$. In such cases, however, the size of the input must also be super-polynomial in $n$, which we can exploit to get reasonable approximations in terms of $n$. This idea underlies the following theorem:

\begin{theorem}
There exists an $n^{2/3}$-approximation algorithm for \ma.
\label{thm:maalg2}
\end{theorem}

\begin{proof}[Proof of Theorem \ref{thm:maalg2}]
The algorithm takes the best of the following:

\begin{enumerate}
    \item The output of the greedy algorithm underlying Lemma \ref{lem:greedyma}.
    \item The best subset of $V$ of size at most $\log_n T$.
    \item The best subset of $V$ generated by the following procedure: Partition $V$ into $r = 2 \lceil \ln T \rceil$ parts $S_1,\ldots,S_{r}$ whose sizes are as close to equal as possible. For each $I \subseteq [r]$, compute the score of the set $\bigcup_{i \in I} S_i$.
\end{enumerate}

As there are $O(n^{\log_n T}) = O(T)$ vertex subsets of size at most $\log_n T$, algorithm 2 runs in $\poly(n,T)$ time. Algorithm 3 also runs in $\poly(n,T)$ time, since we consider $2^{2 \lceil \ln T \rceil} = \poly(T)$ subsets.

When $k \leq \log_n T$, algorithm 2 returns an optimal solution to \ma. Hence it suffices to consider the case that $k \geq \log_n T$. However, in this case, the subset returned by algorithm 3 provides at worst an $\frac{n}{2\ln T}$-approximation to $\opt$. To see this, note that when $V$ is partitioned, the vertices of $\opt$ are split up into parts of size each at most $\tfrac{n}{2\ln T}$. Let $I^* \subseteq [r]$ be the minimal subset such that $\opt \subseteq \bigcup_{i \in I^*} S_i$. Clearly the algorithm considers this set, which induces all the edges induced by $\opt$, and contains at most $|S| \leq \frac{n}{2\ln T} \cdot |\opt|$ vertices.

Let $\hat{S}$ be the best of the solutions returned by these three algorithms. We conclude that

\begin{align*}
    \score(\hat{S}) &\geq \score(\opt) \cdot \max\left(\frac{2\ln T}{n}, \frac{1}{2k \ln T}, \frac{k}{n} \right) \\
    &\geq \score(\opt) \cdot\sqrt[3]{\frac{2\ln T}{n} \cdot \frac{1}{2k \ln T} \cdot \frac{k}{n}} \\
    &= \score(\opt) \cdot \sqrt[3]{\frac{1}{n^2}} \\
    &= \score(\opt) \cdot \frac{1}{n^{2/3}}.
\end{align*}
\end{proof}

\subsection{MinRep-Hardness of \ma}
\label{sec:hardnessma}

We now show that \ma is \minrep-hard to approximate.

\begin{theorem}
There is an approximation-preserving reduction (up to constant factors) from \minrep to \ma. In particular, \ma cannot be approximated to within a $O(2^{\log^{1-\epsilon} n})$ factor unless $\np \subseteq \dtime\left(n^{\mathrm{polylog } n}\right)$.
\label{thm:maminrep}
\end{theorem}

\begin{proof} 
Consider a \minrep instance with supervertices $A_1, A_2, \ldots, A_n$  and $B_1, B_2, \ldots, B_n$, vertex set $\bar{V} = A \cup B$ (with $A = \bigcup_i A_i$ and $B = \bigcup_i B_i$), and edge set $\bar{E}$.  Call an optimal solution to this instance $V^*$. We identify superedges $S_1, S_2, \ldots$  as sets of their constituent edges (so each $S_i$ is a subset of $\bar{E}$). The vertex set in our construction will be $V = \bar{V} \cup \{u, v\}$, where $u$ and $v$ are not vertices in $\bar{V}$.

The frame $G_1$ has exactly one edge, between vertices $u$ and $v$. This forces $u$ and $v$ to each be in $\mathrm{\opt}$, and ensures that $\score(\mathrm{\opt})$ is exactly $\tfrac{1}{|\opt|}$ (all other graphs will contain at least one edge, so $\score(\mathrm{\opt})$ is nonzero). Additionally, it means that we never benefit from picking up more than one edge in any other frame in the sequence.

We now construct one graph for each superedge. For each superedge $S_i \subseteq \bar{E}$, $E_{i+1}$ is exactly $S_i$ (i.e. if $(a_i, b_j) \in S_i$, then $(a_i, b_j) \in E_{i+1}$).  Thus, to get a positive score on $E_{i+1}$, we must pick the endpoints of at least one edge in $S_i$.

The only way to get a positive score is to get a positive score on each frame $G_i$. Thus, for each superedge $S_i$, we must pick both endpoints of one of the edges in $S_i$. Therefore $|\opt| \geq |V^*| + 2$ (where the $+2$ comes from the requirement that we also pick $u$ and $v$). Conversely, picking the vertices in $V^*$ (as well as $u$ and $v$) gives us a positive score in each graph, and we have already established that $\opt$ is precisely the smallest such set. Thus, $|\opt| \leq |V^*| + 2$, meaning that $|\opt| = |V^*| + 2$, and thus $score(\opt) = \frac{1}{|V^*| + 2}$.

Now suppose that it is hard to distinguish between a \minrep instance with objective value $x$ and objective value $x \cdot f(n)$. This implies that it is hard to distinguish between \ma instances with value $\frac{1}{x+2}$ and those with value $\frac{1}{x\cdot f(n) + 2}$. The inapproximability ratio is thus $\frac{\frac{1}{x+2}}{\frac{1}{x\cdot f(n) + 2}} = \frac{x\cdot f(n) + 2}{x + 2} = O(f(n))$.
\end{proof}

\subsection{Planted Dense Subgraph Hardness for Two Frames}
\label{sec:dcs_ma_hard_dks}

Much like \textsc{Planted Clique}, planted instances of \dksf (\dks) are increasingly used to show conditional hardness of approximation for \np-hard problems~\cite{applebaum2010public,arora2010computational,awasthi2015label,charikar2016approximating}. Formally, this \textsc{Planted Dense Subgraph Conjecture} may be phrased as follows. Two graphs, $G_1$ and $G_2$, are independently sampled Erd\H{o}s-R\'{e}nyi random graphs of order $n$ and edge probability $n^{-1/2}$. It is easy to verify that with high probability every size-$\sqrt{n}$ induced subgraph of $G_1$ and $G_2$ has average degree $O(1)$. Subsequently, in $G_2$, a subset of size $\sqrt{n}$ is selected uniformly at random and is replaced with a $G(\sqrt{n}, n^{-1/4-\epsilon})$ Erd\H{o}s-R\'{e}nyi graph for some $\epsilon \in (0, 1/4)$. Thus, the average degree in this subgraph is $\Theta(n^{1/4-\epsilon})$. Finally, $G_3$ is set to equal either $G_1$ or $G_2$, as chosen by a fair coin toss. The \textsc{Planted Dense Subgraph} problem asks, given only access to $G_3$, to determine whether $G_3$ equals $G_1$ (i.e. is \texttt{PLANTED}) or equals $G_2$ (i.e. is \texttt{UNPLANTED}). The \textsc{Planted Dense Subgraph Conjecture} claims that no probabilistic polynomial time algorithm can solve this problem with probability appreciably better than making a random choice (i.e. with probability $1/2 + \epsilon'$).

We show that, assuming this conjecture, \ma has no $n^{1/8-\epsilon}$ approximation even when $T=2$. The intuition behind the construction is as follows. Consider an instance of the \textsc{Planted Dense Subgraph} problem. Because the conjecture (effectively) says that the dense planted component is hard to detect, one might naively imagine that this immediately implies hardness for the \dsf problem (with $T=1$). However, this reasoning is flawed, as the densest subgraph in such instances is (with extremely high probability) simply all of $G$. However, using a second graph $G'$, we can aim to restrict the set of good solutions to those of size at most some $O(k)$. In particular, by setting $G'$ to contain only some $k$-clique, we can ensure that, up to constant factors, the optimum solution contains at most $O(k)$ vertices. Additionally, good solutions to this new problem with $T=2$ directly correspond to good solutions for \dks on $G$. For an appropriate choice of $k$, $n^{1/8-\epsilon}$-hardness follows.

\begin{theorem}
Assuming the \textsc{Planted Dense Subgraph Conjecture}, for no $\epsilon > 0$ is there a probabilistic polynomial-time algorithm approximating \ma, even with only two frames, to within a factor of $O(n^{1/8-\epsilon})$.
\label{thm:dcs_ma_hard_dks18}
\end{theorem}

At the cost of having more frames, we amplify this hardness up to $n^{1/4-\epsilon}$. This is done by reducing from a recursive variant of the Planted Dense Subgraph problem studied by Charikar et. al.~\cite{charikar2016approximating}, allowing us to shrink the relative size of the graph's densest component and thus establish a bigger gap. Although the techniques used are mostly the same as in the proof of Theorem \ref{thm:dcs_ma_hard_dks18}, for ease of presentation we leave the details of our modification to Appendix \ref{sec:dks14} and simply state the result below.

\begin{theorem}
Assuming the \textsc{Planted Dense Subgraph Conjecture}, for no $\epsilon > 0$ is there a probabilistic polynomial-time algorithm approximating \ma to within a factor of $O(n^{1/4-\epsilon})$.
\label{thm:dcs_ma_hard_dks14}
\end{theorem}

\begin{proof}[Proof of Theorem \ref{thm:dcs_ma_hard_dks18}]
Let $G$ be an input to the \textsc{Planted Dense Subgraph} problem, with vertex and edge sets $V$ and $E$. We construct the graph sequence $(G_1, G_2)$ of our \ma instance as follows.

\begin{itemize}
    \item Vertices: Let $U$ be a set of $n^{1/4}$ vertices not in $V$. The vertex set of both $G_1$ and $G_2$ is $V' = V \cup U$.
    \item Edges: $E_1$ contains an edge between every pair of vertices in $U$, and no edges outside of $U$ (thus, $G_1$ is an $n^{1/4}$-clique plus $n$ isolated vertices). $E_2$, on the other hand, is just $E$ (and thus $G_2$ has every vertex in $U$ isolated).
\end{itemize}

We now proceed to prove bounds on $\score(\opt)$. In the following, ``with high probability (w.h.p.)'' means with probability at least $1 - \frac{1}{\poly(n)}$.
\begin{claim}
If $G$ is a \texttt{PLANTED} instance, then $\score(\opt) = \Omega(n^{1/8-\epsilon})$ with high probability.
\label{clm:dcs_planted_good}
\end{claim}

\begin{proof} 
Let $S$ be the union of $U$ and $n^{3/8}$ vertices from the planted component of $G$. Then
\[|E_1[S]| = \binom{n^{1/4}}{2} = \Omega(n^{1/2})\]

and, by a standard application of the Chernoff bounds
\[|E_2[S]| = \Omega\left(n^{-1/4-\epsilon} \cdot \binom{n^{3/8}}{2}\right) = \Omega(n^{1/2-\epsilon})\]
with high probability. Thus,
\[\text{score}(S) = \min\left(\dfrac{|E_1[S]|}{|S|}, \dfrac{|E_2[S]|}{|S|}\right) = \dfrac{\min(\Omega(n^{1/2}), \Omega(n^{1/2-\epsilon}))}{ n^{1/4} + n^{3/8} } = \Omega(n^{1/8-\epsilon}).\]
\end{proof}

\begin{claim}
If $G$ is an \texttt{UNPLANTED} instance, then $\score(\opt) = n^{o(1)}$ with high probability.
\label{clm:dcs_unplanted_bad}
\end{claim}

\begin{proof} 
Let $S$ be an optimal solution, and let $s = |V \cap S|$ be the number of vertices it contains from $V$. Since $\opt \leq \text{score}_{G_2} < s$, we can assume that $s = \Omega(n^{\epsilon'})$ for some $\epsilon' > 0$ (otherwise the proof is trivial). We now bound the quality of $S$ in each of $G_1$ and $G_2$. In $G_1$,
\[\text{score}_{G_1}(S) = \dfrac{|E_1[S]|}{|S|} \leq \dfrac{|E_1[S]|}{s} \leq \dfrac{|E_1|}{s} \leq \dfrac{\sqrt{n}}{s}\]

Meanwhile, in $G_2$, a standard application of Chernoff bounds implies
\[\score_{G_2}(S) = \dfrac{|E_2[S]|}{|S|} \leq \dfrac{|E_2[S]|}{s} \leq \dfrac{O(s^2/\sqrt{n})}{s} = O\left(\dfrac{s}{\sqrt{n}}\right)\]

with high probability. Therefore,
\begin{align*}
    \score(S) &= \min(\score_{G_1}(S), \score_{G_2}(S)) \\
    &\leq \min\left(\sqrt{n}/s),  O(s/\sqrt{n})\right) \\
    &\leq \sqrt{(\sqrt{n}/s) \cdot O(s/\sqrt{n})} \\
    &= O(1) = n^{o(1)}.
\end{align*}
\vspace{-.1in}
\end{proof}

By Claims \ref{clm:dcs_planted_good} and \ref{clm:dcs_unplanted_bad}, the gap between the \texttt{PLANTED} and \texttt{UNPLANTED} cases is at least
\[\dfrac{\Omega(n^{1/8-\epsilon})}{O(n^{o(1)})} \geq \Omega(n^{1/8 - \epsilon'}).\]
with high probability. This concludes the proof of Theorem \ref{thm:dcs_ma_hard_dks18}.
\end{proof}

\subsection{Integrality Gap Example}

In \cite{jethava2015finding}, the authors introduce the following linear program (DCS\_LP) for \ma.

\begin{align*}
	\textsc{maximize\,\,\,\,}\hspace{.1in}&\hspace{.05in} z
	\\  
	\textsc{subject to}\hspace{.2in} &\textstyle \sum_i y_i = 1 \\
	& x_{(u,v)} \leq y_u & \textrm{for all } u, v \in V \\
	& x_{(u,v)} \leq y_v & \textrm{for all } u, v \in V \\
	& z \leq \textstyle \sum_{e \in E_t} x_{e,t} & \textrm{for all } t \in T \\
	& \mathbf{x}, \mathbf{y} \geq \mathbf{0}
\end{align*}
	
When $T = 1$, this simplifies to the LP shown by \cite{charikar2000greedy} to solve the \dsf problem exactly. When $T$ is allowed to be  larger than $1$, we show that the integrality gap of this LP can be near-linear.\footnote{We use the term ``integrality gap'' loosely here. As the LP is normalized, the intended solution has $y_v$ set to $1/|S|$ when $v$ is in the optimal solution $S$ and to $0$ otherwise. What we measure is really the ratio of LP-OPT to the score of true optimal feasible solution for the given instance.}

\begin{theorem}
DCS\_LP has an integrality gap of $\Omega\left(\tfrac{n}{\log n}\right)$.
\label{thm:lpgap}
\end{theorem}

\begin{proof}
Label $n$ vertices $1$ through $n$, and consider the instance composed of the following sequence of $T = n-1$ graphs, $G_1$ through $G_{n-1}$. $G_1$ contains a single edge from $v_1$ to $v_2$.  $G_2$ contains two edges: one from $v_1$ to $v_3$ and the other from $v_2$ to $v_3$. In general, $G_k$ contains $k$ total edges, each with one endpoint at $v_{k+1}$ and the other at $v_i$ for $i \in [k]$.

We first consider the optimal integral solution to this graph sequence. Because each vertex is the center of a star in at least one of the graphs ($G_1$ has both $v_1$ and $v_2$ as centers), not picking even one of the vertices ensures that the corresponding graph attains average degree $0$, and thus the minimum average degree for the sequence is also $0$. Thus, to get any positive objective value, we must pick every vertex. Consequently, because $G_1$ contains a single edge, the objective value of our solution is just $1/n$.

Now we consider what the LP can achieve. Set $h = \tfrac{1}{1 + H_{n-1}}$ (where $H_k = \sum_{i=1}^k \tfrac{1}{i}$), and consider the fractional solution that assigns $y_1 = h$ and  $y_i = \tfrac{h}{i-1}$ for $i \geq 2$. The sum of these is $h\cdot \left(1 + \sum_{i=2}^n \tfrac{1}{i-1}\right) = h \cdot \tfrac{1}{h} = 1.$ Note that these values are monotonically nonincreasing in $i$, and thus if $i \geq j$, then we can set $x_{(v_i,v_j)} = \min(y_{v_i}, y_{v_j}) = y_{v_i} = \tfrac{h}{i-1}$.  For each $k$, graph $G_k$ contains $k$ edges between $v_{k+1}$ and vertices with a smaller index. Thus, our assignment induces exactly $k \cdot \frac{h}{k+1-1} = h$ fractional edges in each $G_k$, meaning that LP-OPT is at least $h = \tfrac{1}{H_{n-1}+1} = \Omega\left({1}/{\log n}\right)$.  Thus, the integrality gap is $\frac{\Omega\left(1/ \log n\right)}{1/n} = \Omega\left(\frac{n}{\log n}\right)$.
\end{proof}

\section{Algorithms and Hardness for \am}
\label{sec:am}

\subsection{Hardness of \am}
\label{sec:hardnessam}

We now show that \am is \np-hard even to approximate to any significantly nontrivial factor.

\begin{theorem}
\am has no $n^{1-\epsilon}$-approximation algorithm for any $\epsilon > 0$ unless $\p = \np$.
\label{thm:MIS_to_AM}
\end{theorem}

\begin{proof}
We show this by presenting a direct reduction from \misf (\mis), which is well-known to have the aforementioned hardness factor~\cite{hastad1999clique, zuckerman2006linear}. Given an \mis instance with a graph $G=(V,E)$ that is not complete (the problem is trivial otherwise), we construct a \am instance consisting of one frame $G_v$ for each vertex $v \in V$. In each such frame $G_v$, all of $v$'s neighbors are singletons, while the remaining vertices form a star centered at $v$. We now show that the size of the maximum independent set in $G$ is equal to the maximum feasible objective in the constructed \am instance.

Suppose $I \subseteq V$ is an independent set of size $k \geq 2$ in $G$. Then consider the solution $I$ to the constructed \am instance. For each $v \in I$, we score $1$ point in frame $G_v$, since $v$ has a neighbor in $G_v$ (some other vertex in $I$) and none of the singletons of $G_v$ (the neighbors of $v$ in $G$) are in $I$. Thus, in the \am instance, $\text{score}(I) \geq k$, and thus $\opt(\am) \geq \opt(\mis)$.

In the reverse direction, suppose that $S \subseteq V$ is a solution to the constructed \am instance achieving $\text{score}(S) = k$. It is easy to check that we can only score at most one point per frame; thus there must be $k$ frames in which $S$ induces nonzero min degree. Consider any such frame $G_v$, corresponding to $v \in V$. Since our score in $G_v$ is $1$, we must have $v \in S$, as it is the center of the star in $G_v$. And we cannot have chosen any of $v$'s neighbors in $G$, for otherwise there would be a singleton in $G_v[S]$. We conclude that $S$ contains at least $k$ vertices, no pairs of which are neighbors in $G$; i.e. $S$ contains an independent set of size $k$ (precisely those vertices at the centers of frames in which we scored). Thus, $\opt(\mis) \geq \opt(\am)$.

Therefore, the optimal objective value of the constructed \am instance exactly equals that of the given \textsc{MIS} instance, so \am has a hardness factor at least as large as that of \textsc{MIS}.
\end{proof}

\subsection{Fixed-Parameter Algorithms for \am}
\label{sec:fixedparameteram}

In light of the above hardness result, we now direct our attention to fixed-parameter algorithms. In particular, we show how to solve \am for small $T$ by generalizing classical algorithms for finding $k$-cores in a graph. Concretely, we provide an $n^T \cdot \poly(n,T)$-time algorithm for the exact version of \am, as well as an $f(T) \cdot \poly(n,\epsilon^{-1})$-time (i.e. FPT-time) $(1+\epsilon)$-approximation algorithm for some computable function $f$.

Given a graph sequence $\mathcal{G} = (G_1, G_2, \cdots, G_T)$ over vertices $V$, we say that a set $S \subseteq V$ is a \emph{$(k_1,\cdots,k_T)$-core} if it induces minimum degree at least $k_i$ in each frame $G_i$. In other words, for all $i \in [T]$, $S$ satisfies $\text{min-deg}(G_i[S]) \geq k_i$.

A $(k_1,\cdots,k_T)$-core can be computed in $\text{poly}(n,T)$ time if one exists, via a simple algorithm described in two recent works~\cite{azimi2014k,galimberti2017core}; we include it here for completeness. Starting from a set $S$ containing all of $V$, we repeatedly remove from $S$ any vertex whose degree in $G_i$ is less than $k_i$. When no such vertices remain, we return $S$. The returned set is either empty or the desired core. This works because if a vertex $v$ is deleted, then by definition it cannot be part of any $(k_1,\cdots,k_T)$-core of a graph induced by a subset of $S$.

We can now use this procedure as a black box to derive the following two results.

\begin{theorem} There is an exact algorithm for \am with running time $n^T \cdot \poly(n,T)$.
\label{thm:dcs_am_alg_exact}
\end{theorem}

\begin{theorem} For some computable function $f$ and every $\epsilon > 0$, there is a $(1+\epsilon)$-approximation algorithm for AM with running time $f(T) \cdot \poly(n, \epsilon^{-1})$ (i.e, the algorithm is fixed parameter tractable in $T$).
\label{thm:dcs_am_alg_approx}
\end{theorem}

\begin{proof}[Proof of Theorem \ref{thm:dcs_am_alg_exact}]
This algorithm simply returns the largest integer $k$ such that $\mathcal{G}$ has a $(k_1,\cdots,k_T)$-core with $\sum_{i=1}^T k_i = k$. Since there are at most $n^T$ tuples of the form $(k_1,\cdots,k_T)$, this runs in $n^T \cdot \poly(n,T)$ time.
\end{proof}

\begin{proof}[Proof of Theorem \ref{thm:dcs_am_alg_approx}]
As before, we intend to return the largest integer $k$ such that $\mathcal{G}$ has a $(k_1,\cdots,k_T)$-core for some $\sum_{i=1}^T k_i = k$. However this time, we only consider those tuples $(k_1,\cdots,k_T)$ such that $k_i = (1+\epsilon)^{\ell_i}$ for some integers $\ell_1,\cdots,\ell_T$. As one of the solutions considered contains the optimal solution with the corresponding vector rounded down to the nearest power of $1+\epsilon$ (and thus each entry of the vector is within a $(1+\epsilon)$ factor of those in $\opt$), the sum of the entries in some rounded-down vector is within a $(1+\epsilon)$ factor of the value of the optimal solution.

The upshot is that we have a $(1+\epsilon)$-approximation algorithm that only considers $O\left((\log_{1+\epsilon}{n})^T\right) = O((\log n)^T / \epsilon)$ different $k$-cores. Thus, the total running time is $(\log n)^T \cdot \poly(n, T, \epsilon^{-1})$. By AM-GM,
\[(\log n)^T = 2^{T \log\log n} \leq 2^{\frac{T^2 + \log\log^2 n}{2}} = 2^{T^2/2} \cdot 2^{(\log\log^2 n) / 2} = 2^{T^2/2} \cdot n^{o(1)}\]

Therefore, $(\log n)^T \cdot \poly(n, T, \epsilon^{-1}) = f(T) \cdot \poly(n, T, \epsilon^{-1}) = f(T) \cdot \poly(n,\epsilon^{-1}).$
\end{proof}


\bibliography{references}

\newcommand{\noop}[1]{}
\begin{thebibliography}{10}

\bibitem{anderssonfinding}
Alexander Reinthal Anton T{\"o}rnqvist~Arvid Andersson and Erik Norlander
  Philip Stalhammar~Sebastian Norlin.
\newblock Finding the densest common subgraph with linear programming.
\newblock {\em Manuscript}, pages 1--34, 2016.

\bibitem{applebaum2010public}
Benny Applebaum, Boaz Barak, and Avi Wigderson.
\newblock Public-key cryptography from different assumptions.
\newblock In {\em Proceedings of the forty-second ACM symposium on Theory of
  computing}, pages 171--180. ACM, 2010.

\bibitem{arora2010computational}
Sanjeev Arora, Boaz Barak, Markus Brunnermeier, and Rong Ge.
\newblock Computational complexity and information asymmetry in financial
  products.
\newblock In {\em ICS}, pages 49--65, 2010.

\bibitem{awasthi2015label}
Pranjal Awasthi, Moses Charikar, Kevin~A. Lai, and Andrej Risteski.
\newblock Label optimal regret bounds for online local learning.
\newblock In {\em Proceedings of the 28th Conference on Learning Theory
  (COLT)}, pages 150--166, 2015.

\bibitem{axiotis2016size}
Kyriakos Axiotis and Dimitris Fotakis.
\newblock On the size and the approximability of minimum temporally connected
  subgraphs.
\newblock In {\em LIPIcs-Leibniz International Proceedings in Informatics},
  volume~55. Schloss Dagstuhl-Leibniz-Zentrum fuer Informatik, 2016.

\bibitem{azimi2014k}
N~Azimi-Tafreshi, J~G{\'o}mez-Gardenes, and SN~Dorogovtsev.
\newblock $k$-core percolation on multiplex networks.
\newblock {\em Physical Review E}, 90(3):032816, 2014.

\bibitem{charikar2000greedy}
Moses Charikar.
\newblock Greedy approximation algorithms for finding dense components in a
  graph.
\newblock In {\em International Workshop on Approximation Algorithms for
  Combinatorial Optimization}, pages 84--95. Springer, 2000.

\bibitem{charikar2016approximating}
Moses Charikar, Yonatan Naamad, and Anthony Wirth.
\newblock On approximating target set selection.
\newblock In {\em LIPIcs-Leibniz International Proceedings in Informatics},
  volume~60. Schloss Dagstuhl-Leibniz-Zentrum fuer Informatik, 2016.

\bibitem{charikar2017label}
Moses Charikar, Yonatan Naamad, and Anthony\noop{} Wirth.
\newblock On {D$k$S} hardness for {MinRep}-hard problems.
\newblock {\em Manuscript}, pages 1--14, 2017.

\bibitem{dinur2005new}
Irit Dinur, Venkatesan Guruswami, Subhash Khot, and Oded Regev.
\newblock A new multilayered {PCP} and the hardness of hypergraph vertex cover.
\newblock {\em SIAM Journal on Computing}, 34(5):1129--1146, 2005.

\bibitem{galimberti2017core}
Edoardo Galimberti, Francesco Bonchi, and Francesco Gullo.
\newblock Core decomposition and densest subgraph in multilayer networks.
\newblock In {\em Proceedings of the 2017 {ACM} on Conference on Information
  and Knowledge Management, {CIKM} 2017, Singapore, November 06 - 10, 2017},
  pages 1807--1816, 2017.
\newblock URL: \url{http://doi.acm.org/10.1145/3132847.3132993}, \href
  {http://dx.doi.org/10.1145/3132847.3132993}
  {\path{doi:10.1145/3132847.3132993}}.

\bibitem{goldberg1984finding}
Andrew~V Goldberg.
\newblock {\em Finding a maximum density subgraph}.
\newblock University of California Berkeley, CA, 1984.

\bibitem{hu2005mining}
Haiyan Hu, Xifeng Yan, Yu~Huang, Jiawei Han, and Xianghong~Jasmine Zhou.
\newblock Mining coherent dense subgraphs across massive biological networks
  for functional discovery.
\newblock {\em Bioinformatics}, 21(suppl\_1):i213--i221, 2005.

\bibitem{hastad1999clique}
Johan Håstad.
\newblock Clique is hard to approximate within $n^{1-\epsilon}$.
\newblock {\em Acta Math.}, 182(1):105--142, 1999.
\newblock URL: \url{http://dx.doi.org/10.1007/BF02392825}, \href
  {http://dx.doi.org/10.1007/BF02392825} {\path{doi:10.1007/BF02392825}}.

\bibitem{jethava2015finding}
Vinay Jethava and Niko Beerenwinkel.
\newblock Finding dense subgraphs in relational graphs.
\newblock In {\em Joint European Conference on Machine Learning and Knowledge
  Discovery in Databases}, pages 641--654. Springer, 2015.

\bibitem{jethava2013lovasz}
Vinay Jethava, Anders Martinsson, Chiranjib Bhattacharyya, and Devdatt
  Dubhashi.
\newblock Lov{\'a}sz $\vartheta$ function, {SVMs} and finding dense subgraphs.
\newblock {\em The Journal of Machine Learning Research}, 14(1):3495--3536,
  2013.

\bibitem{kortsarz2001hardness}
Guy Kortsarz.
\newblock On the hardness of approximating spanners.
\newblock {\em Algorithmica}, 30(3):432--450, 2001.

\bibitem{semertzidis2016best}
Konstantinos Semertzidis, Evaggelia Pitoura, Evimaria Terzi, and Panayiotis
  Tsaparas.
\newblock Best friends forever {(BFF)}: Finding lasting dense subgraphs.
\newblock {\em arXiv preprint arXiv:1612.05440}, pages 1--15, 2016.

\bibitem{zuckerman2006linear}
David Zuckerman.
\newblock Linear degree extractors and the inapproximability of max clique and
  chromatic number.
\newblock In {\em Proceedings of the thirty-eighth annual ACM symposium on
  Theory of computing}, pages 681--690. ACM, 2006.

\end{thebibliography}


\appendix

\section{Proof of \texorpdfstring{$n^{1/4-\epsilon}$}{n\textasciicircum1/4-eps} Hardness for \texorpdfstring{\kma}{Densest k-Common Subgraph} from Planted DkS}
\label{sec:dks14}

A critical tool in this section will be the analysis of the \textsc{Recursive Planted Dense Subgraph} problem first studied in~\cite{charikar2016approximating}. The statement of this problem takes two equal-length vectors, the \emph{size vector} $\vec{n} = (n_1, n_2, \cdots n_r)$ and the \emph{log-density vector} $\vec{d} = (p_1, p_2, \cdots p_r)$ as parameters. Much like in the standard \textsc{Planted Dense Subgraph} problem, inputs are then sampled from either a planted or unplanted distribution. In the unplanted case, the returned graph is simply a $G(n_1, p_1)$ Erd\H{o}s-R\'enyi random graph. In the planted case, the distribution $\mathcal{D}(\vec{n}, \vec{p})$ of returned graph is constructed recursively as follows: 

  \[
    \mathcal{D}(\vec{n},\vec{p})=
    \left\{
    \begin{array}{llr}
      G(n_1, {n_1}^{p_1 - 1})&& \text{if } r=1\\
      G(n_1, {n_1}^{p_1 - 1})& \text{planted with a graph sampled} & \text{if } r>1 \\
&\text{from }\mathcal{D}((n_2, n_3, \cdots, n_r), (p_2, p_3, \cdots, p_r))  &
    \end{array}
    \right.
  \]
  
Here, $G$ being ``planted with'' $H$ means that a randomly-chosen subgraph of $G$ of order $|V(H)|$ has a copy of $H$'s edges unioned into its current induced edge set (thus, this construction only makes sense for monotonically decreasing vectors $\vec{n}$). In particular, we reconstruct the un-recursed form of \textsc{Planted Dense Subgraph} when $\vec{n} = (n, \sqrt{n})$ and $\vec{p} = (\tfrac{1}{2}, \tfrac{1}{2} - \epsilon)$. As we increase the length of the two parameter vectors, we get additional ``layers'' of planting, with each layer included in the previous.

A priori, one might expect that the problem of distinguishing the two distributions becomes easier as the number of rounds of planting increases (as additional planting can only ever increase the density of all extant dense planted components). A lemma central to~\cite{charikar2016approximating}, however, states that for a carefully chosen parameter sequence, the problem may remain just as intractable as the two-layer (``un-recursed'') problem. In particular, adapted to our use case, Lemma~2 of that paper effectively states the following:

\begin{lemma}[Lemma 4 of~\cite{charikar2016approximating}]
For $r>2$, let $\vec{n} = (n, n^{1/2}, n^{1/4}, \cdots, n^{1/2^{r}})$ and let $\vec{p} =$

\mbox{$(\tfrac{1}{2}, \tfrac{1}{2}-\epsilon_1, \tfrac{1}{2}-\epsilon_2, \cdots, \tfrac{1}{2}-\epsilon_r)$}, where $ \epsilon_{i+1} \in (\tfrac{\epsilon_{i}}{2}, \epsilon_{i})$ for $i < r$. Assuming the Planted Dense Subgraph conjecture, there is no probabilistic polynomial-time algorithm for the Recursed Planted Dense Subgraph problem with parameters $\vec{n}$ and $\vec{p}$.
\end{lemma}

We now proceed to use this lemma to amplify the hardness obtained in Theorem \ref{thm:dcs_ma_hard_dks18} up to $n^{1/4-\epsilon}$. We first exhibit a straightforward but fallacious approach, and then show how to correct its flaw.

\begin{proof}[Fallacious proof of Theorem \ref{thm:dcs_ma_hard_dks14}]

Let $G=(V,E)$ be an input to the recursed Planted Dense Subgraph problem with $r$ rounds of planting in which all log densities lie within the interval $(1/2-\epsilon, 1/2]$. The construction of this ``proof'' begins much like that of Theorem \ref{thm:dcs_ma_hard_dks18}.

\subparagraph*{Vertices}
Let $U$ be a set of $n^{1/2^{r}}$ vertices not in $V$. The vertex set of both $G_1$ and $G_2$ is $V' = V \cup U$.
\subparagraph*{Edges}
$E_1$ contains an edge between every pair of vertices in $U$, and no edges outside of $U$ (and thus, $G_1$ is an $n^{1/2^{r}}$-clique plus $n$ isolated vertices). $E_2$, on the other hand, is just $E$ (and thus $G_2$ has every vertex in $U$ isolated).

One can now try to use the same argument as in Claim \ref{clm:dcs_planted_good} and attempt to derive a tightened version of Claim \ref{clm:dcs_unplanted_bad} to show that both
\begin{enumerate}
    \item In the planted case, picking $U$ plus the vertices of $V$ from the innermost planted component of $G$ certifies that $\score(\opt) = \Omega(n^{-\epsilon'})$ for some $\epsilon' = \Theta(\epsilon)$ and
    \item In the unplanted case, the best we can do is pick $U$ plus an arbitrary size-$|U|$ subset of $V$, so $\score(\opt) = O(n^{-1/4 + 1/2^{r}})$.
\end{enumerate}

Unfortunately, the second of these is false, as is exhibited by the trivial solution selecting the endpoints of one edge from each of the two frames, which gets a score of $\Omega(1)$. In particular, we do not have a good lower bound on $s$, so we cannot usefully apply the Chernoff bounds. Additionally, this shows that the $\Omega(n^{-\epsilon'})$ lower bound is both trivial and unhelpful. Since an $\Omega(1/T)$ algorithm is trivial, we need our Yes instances to have a score of at least $\Omega(n^{1/4-\epsilon'}/T)$ if we ever want to achieve the sought bounds. As we now show, one way to do this is by increasing $T$.

Our construction will be exactly as above, except we pad the graph sequence with an additional $n^2$ different i.i.d. $G(n, n^{-3\epsilon'})$ random graphs on vertex set $V$. For any fixed $O(n^{\epsilon'})$-sized subset of $V$, the probability that it induces an edge in one of the subgraphs is $O(n^{-\epsilon})$, and the probability that it induces an edge in \textit{all} of them is $O(n^{-n^2\epsilon})$. Because there are only $O(n^{n^\epsilon})$ many such subsets, the union bound ensures that with high probability every vertex subset of size $O(n^{\epsilon})$ misses an edge in at least  one graph in the sequence, and thus has score $0$. Therefore, we know that $\opt > n^{\epsilon'}$, which allows us to use Chernoff bounds as in Claim \ref{clm:dcs_unplanted_bad}. Additionally, also by Chernoff bounds, solutions of size $\Omega(n^{4\epsilon'})$ have their objective scores simply scaled down by a factor of $n^{2\epsilon'}$ (up to subconstant factors and w.h.p.), so the relative value of all large solutions remains unchanged.

Thus, with this additional change, both arguments (a) and (b) above hold (up to factors of $n^{\Theta(\epsilon')}$), and we establish a gap of $n^{1/4 - 1/2^{r+1} - \Theta(\epsilon')}$. Rewriting the exponent and for a suitable choice of $r$ (dependent on $\epsilon$), our gap can be set to equal $n^{1/4 - \epsilon}$.
\end{proof}

\section{\texorpdfstring{\dksf}{Densest k-Subgraph} Hardness for \texorpdfstring{\kma}{Densest k-Common Subgraph}}
\label{sec:kma}

In \cite{semertzidis2016best}, Semertzidis et al. also study a generalization of \ma in which the score of a solution $S \subset V$ is the average degree of the $k$th-densest subgraph induced by $S$ in the sequence. \ma is simply the restriction of this problem to $k = T$. We now argue that this problem, \kma, has hardness related to that of \dksf (\dks). Namely, we show that if \kma has an $f(n)$-approximation algorithm, then \dks can be approximated to within a factor of $O(f(n)^2)$.

\begin{theorem}
If \kma has an $f(n)$-approximation algorithm, then \dks has an $O(f(n)^2)$-approximation algorithm.
\label{thm:dkshard}
\end{theorem}

\begin{proof}
In \cite{charikar2017label}, the authors show a similar \dks hardness for $k$-\minrep, the generalization of \minrep in which feasible solutions need to cover only $k$ total superedges (as opposed to all of them). Given an instance of $k$-\minrep, we produce a \kma instance using the same reduction as in the proof of Theorem \ref{thm:maminrep}. The choice of $k$ remains the same between the two problems. Using the same analysis as before, we get that the \kma problem is as hard to approximate as $k$-\minrep (up to constant factors). The conclusion follows.
\end{proof}

\section{Common Spanning Subgraphs}
\label{sec:mcss}

In this section, we consider the natural extension of the \textsc{Minimum Spanning Tree (MST)} problem to sequences of graphs. Just as \textsc{MST} is often motivated by the design of communication networks, one can imagine that the network links are known to change over a set of discrete times, and the goal is to purchase a minimal set of links to ensure that every node is always connected to the rest. We note that a similar but distinct problem, ``Minimum Temporally Connected Subgraphs'', was studied recently by Axiotis and Fotakis \cite{axiotis2016size}.

\begin{definition}[\mcssf (\mcss)]
Given a sequence of connected graphs $G_1=(V,E_1), \ldots, G_T=(V,E_T)$, find a minimum-size set of edges $E^* \subseteq \bigcup_t E_t$ that induces a spanning subgraph in every frame.
\end{definition}

\subsection{Approximating \mcss}
\label{sec:approx_mcss}

Here we consider a natural greedy algorithm for \mcss. Suppose we are given an instance $G_1=(V,E_1), \ldots, G_T=(V,E_T)$, with $E = \bigcup_t E_t$. We build up a solution, starting with empty graphs $H_1=(V,\emptyset),\ldots,H_T=(V,\emptyset)$. While the total number of connected components in all these graphs is greater than $T+n$, pick the edge $e \in E$ that reduces the total number of connected components by the greatest amount. Finally, once there are only $T+n$ connected components in total, add any $n$ edges that bring the number of connected components down to $T$.

\begin{theorem}
The above greedy procedure is a $O(\log T)$-approximation algorithm for \mcss.
\end{theorem}

\begin{proof}
We analyze the algorithm in terms of a \emph{potential function} $\rho_i$. On the $i$-th iteration---that is, the $i$-th edge picked---the potential is defined as
\[
    \rho_i = \sum_{t=1}^T (\text{number of connected components in } H_t) - T
\]
Initially $\rho_0 = nT - T$, since every vertex is a singleton. Once the potential is reduced to at most $n$, there are at most $T+n$ connected components and the while-loop ends.

Consider some $i$-th iteration. Since adding the optimal solution would drop the potential from $\rho_i$ down to $0$, there exists an edge whose addition would decrease the potential by at least $\rho_i/\opt$. After adding the greedily-chosen edge, which is at least as good, we have a potential of
\[
    \rho_{i+1} \leq \rho_i - \rho_i/\opt = \rho_i(1 - 1/\opt).
\]
Consequently after $j$ iterations we have
\[
    \rho_j \leq \rho_0 \left( 1 - \frac{1}{\opt} \right)^j \leq (nT-T) e^{-j/\opt}
\]
A choice of $j = \ln(T) \cdot \opt$ iterations suffices to achieve $\rho_j < n$. Finally, in its last step, the algorithm adds at most $n$ edges to reduce the number of connected components down to $T$, yielding a feasible solution. But since any feasible solution must span at least one frame, $\opt \geq n-1$; therefore our solution has size at most $\ln(T) \cdot \opt + \opt + 1 = O(\log T) \cdot \opt$.
\end{proof}

\subsection{Hardness of Approximating \mcss}
\label{sec:hardness_mcss}

\begin{theorem}
\mcss is \np-hard to approximate to within a factor of $2-\epsilon$ for every $\epsilon>0$.
\label{thm:mcss_hardness}
\end{theorem}

To prove this, we first give a generic reduction from \setcover, which shows a gap hardness in terms of certain parameters. Later we show how reducing from a special case of \setcover yields parameters that give the desired gap.

\begin{lemma}
Given a \setcover instance on $m$ subsets, we can construct in polynomial time an \mcss instance that has a solution of size $m+c+1$ iff the constructed \setcover instance has a solution of size $c$.
\label{lem:setcover_to_MCSS}
\end{lemma}

\begin{proof} 
Given a set system $\mathcal{S} = \{S_1,\ldots,S_m\}$ over elements $U = \{x_1,\ldots,x_n\}$, we create the following sequence of unweighted graphs $G_0=(V,E_0), G_1=(V,E_1), \ldots, G_n=(V,E_n)$:

\begin{itemize}
    \item Every graph has the same vertex set $V = \{s_1,\ldots,s_m,x,y\}$.
    \item $E_0$ forms the path $x,y,s_1,\ldots,s_m$.
    \item For each $1 \leq i \leq n$, $E_i$ contains the edges along the path $y,s_1,\ldots,s_m$, and, for each set $S_j \ni x_i$, the edge $(s_j,x)$.
\end{itemize}

Suppose there is a \setcover solution $\mathcal{T} \subseteq \mathcal{S}$ of size $c$. Consider the \mcss solution $F \subseteq E$ consisting of:
\begin{itemize}
    \item The edges along the path $x,y,s_1,\ldots,s_m$.
    \item For each $S_j \in \mathcal{T}$, the edge $(s_j,x)$.
\end{itemize}
This is precisely $m+1+c$ edges. To see that $F$ is a valid \mcss solution, observe that:
\begin{itemize}
    \item All of $E_0$ is picked, so $G_0[F]$ is connected.
    \item For $1 \leq i \leq n$, the edges along the path $y,s_1,\ldots,s_m$ are in $F$. Additionally, by virtue of being a set cover, there is a set $S_j \in \mathcal{T}$ that contains $x_i$; therefore $F$ contains the edge $(s_j,x)$, ensuring that $x$ is connected to the rest of the frame.
\end{itemize}

Conversely, suppose there is an \mcss solution $F$ of size $m+c+1$. First, observe that $F$ must contain the $m+1$ edges along the path $x,y,s_1,\ldots,s_m$, for otherwise $G_0[F]$ would not be connected. Hence $F$ has exactly $c$ edges outside of this path. The only such edges are of the form $(s_j,x)$. Consider the \setcover solution $\mathcal{T} = \{S_j : (s_j,x) \in F\}$. Clearly this is of size $c$. And it is a valid solution, since we pick for each element $x_i$ at least one set $S_j \ni x_i$, which corresponds to a vertex $s_j$ with an edge to $x$ in $G_i$.
\end{proof}

Now, using the above construction, we reduce from the following special case of \setcover:

\begin{definition}[\ekvcf (\ekvc)]
Given a $k$-uniform hypergraph, pick a minimum-size set of vertices so that every hyperedge has an endpoint in the set.
\end{definition}

Variants of the following result are proven in several papers; see for example~\cite{dinur2005new}.

\begin{theorem}
For every $k \geq 2, \varepsilon>0$ it is \np-hard to distinguish, given an instance of \ekvc on an $n$-vertex, $k$-uniform hypergraph, the following cases:
\begin{itemize}
    \item (YES) There is a vertex cover of size at most $\left( O(\frac{1}{k}) + \varepsilon \right)n$.
    \item (NO) Every vertex cover has size at least $(1 - \varepsilon)n$.
\end{itemize}
\end{theorem}

We can now easily show the desired gap hardness for \mcss.

\begin{proof}[Proof of Theorem \ref{thm:mcss_hardness}]
We reduce from \ekvc. Suppose we are given a hypergraph $\mathcal{G}=(\mathcal{V},\mathcal{E})$ where $|\mathcal{V}|=n$. Since \ekvc is a special case of \setcover, we can perform the reduction described in Lemma \ref{lem:setcover_to_MCSS}. Naturally, for each vertex there is a subset consisting of its incident edges, and the universe of elements is $\mathcal{E}$.

Combining the gap hardness of \ekvc with Lemma \ref{lem:setcover_to_MCSS} shows that it is \np-hard to distinguish:
\begin{itemize}
    \item (YES) There is a solution of size at most $n + \left( O(\frac{1}{k}) + \varepsilon \right)n + 1$.
    \item (NO) Every solution has size at least $n + (1-\varepsilon)n + 1$.
\end{itemize}

This gives a gap of
\[
    \dfrac{\text{NO}}{\text{YES}} \geq \dfrac{(2-\varepsilon)n + 1}{(1 + O(1/k) +\varepsilon)n + 1} \geq 2 - \epsilon
\]
for an appropriate choice of $\varepsilon, k$ in terms of $\epsilon$.
\end{proof}

\end{document}